\pdfoutput=1
\documentclass{article}

\usepackage[utf8]{inputenc}
\usepackage[T1]{fontenc}
\usepackage{amsmath,amssymb,amsthm}
\usepackage{graphicx}
\usepackage{algorithm}
\usepackage{algpseudocode}
\usepackage{float}
\usepackage{geometry}
\usepackage{hyperref}
\usepackage{orcidlink}

\newtheorem{Theorem}{Theorem}[section]
\newtheorem{Definition}{Definition}[section]
\newtheorem{proposition}{Proposition}[section]
\theoremstyle{remark}

\title{Entropy--Rank Ratio: A Novel Entropy--Based Perspective for DNA Complexity and Classification}

\author{%
Emmanuel Pio Pastore\orcidlink{0009-0007-7851-4414}\textsuperscript{1},
Giuseppe Passarino\orcidlink{0000-0003-4701-9748}\textsuperscript{1},
Peppino Sapia\orcidlink{0000-0003-3794-3232}\textsuperscript{2},
Francesco De Rango\orcidlink{0000-0002-2328-8487}\textsuperscript{1}\thanks{Corresponding author: francesco.derango@unical.it}\\[0.75em]
\textsuperscript{1}Department of Biology, Ecology and Earth Science\\
University of Calabria, 87036 Rende, Italy\\
\textsuperscript{2}Department of Mathematics and Computer Science\\
University of Calabria, 87036 Rende, Italy
}

\date{November 06 2025}

\begin{document}

\maketitle

\fontsize{11.5}{14}\selectfont
\sloppy

\begin{abstract}
Shannon entropy is widely used to measure the complexity of DNA sequences but suffers from saturation effects that limit its discriminative power for long uniform segments. We introduce a novel metric, the \emph{entropy rank ratio} $R$, which positions a target sequence within the full distribution of all possible sequences of the same length by computing the proportion of sequences that have an entropy value equal to or lower than that of the target. In other words, $R$ expresses the relative position of a sequence within the global entropy spectrum, assigning values close to 0 for highly ordered sequences and close to 1 for highly disordered ones. DNA sequences are partitioned into fixed-length subsequences and non-overlapping $n$-mer groups; frequency vectors become ordered integer partitions and a combinatorial framework is used to derive the complete entropy distribution. Unlike classical measures, \(R\) is a normalized, distribution-aware measure bounded in \([0,1]\) at fixed \((T,n)\), which avoids saturation to \(\log_2 4\) and makes values comparable across sequences under the same settings. We integrate $R$ into data augmentation for convolutional neural networks by proposing ratio-guided cropping techniques and benchmark them against random, entropy-based, and compression-based methods. On two independent datasets, viral genes and human genes with polynucleotide expansions, models augmented via $R$ achieve substantial gains in classification accuracy using extremely lightweight architectures.
\end{abstract}

\maketitle

\section{Introduction}

DNA is the fundamental molecule of life and is found in all living organisms. It represents a code formed by four distinct letters: the nitrogenous bases (Adenine, Cytosine, Guanine, Thymine). These bases, by combining in sequences of varying lengths and complexity, give rise to highly variable genetic information\cite{ref1}. In fact, it has been known since the last century that DNA sequences encoding proteins are read in groups of three consecutive letters; in total, 64 triplets encode 20 amino acids (as well as the start and stop signals of the code)\cite{ref2}.

Since the discovery of this coding system, efforts have been made to represent the information and complexity of DNA sequences in a mathematical and information-theoretic perspective\cite{ref3,ref4}. Since the 1940s, mathematical and algebraic techniques have been applied to this subject, but it was only after the discovery of the structure of DNA and the introduction of Shannon entropy that it became theoretically possible to calculate the entropy of DNA sequences\cite{ref4,ref5,ref6}. This development sparked strong interest in the topic during the 1980s and 1990s\cite{ref4,ref6,ref7}, aided by the increasing availability of computational resources. However, several challenges arise when calculating the Shannon entropy of sequences: without additional precautions, the entropy of sufficiently long sequences converges to 2 bits, and the practical applications of entropy itself are not immediately evident.

Entropy-based techniques have, however, been widely used in Machine Learning for a long time\cite{ref8,ref9}, and have also been applied in DNA sequence classifiers\cite{ref10,ref11}. Neural networks commonly used in this regard are CNNs (Convolutional Neural Networks), which capture local sequence patterns and aggregate them into global representations\cite{ref12,ref13}.

The objective of this paper is to develop a new method for characterizing the complexity of DNA sequences and to integrate it with the aforementioned neural network architectures in order to significantly improve precision and performance on datasets of viral genes from a GitHub repository\cite{ref14}, and of human genes associated with polynucleotide expansion\cite{ref15}.

\section{Materials and Methods}

DNA can be viewed as an alphabet \cite{ref6,ref16,ref17,ref18} $\mathcal{A}$ consisting of $\lambda_1 = 4$ different single-digit letters 
\begin{equation}
\mathcal{A} = \{\mathrm{A}, \mathrm{C}, \mathrm{G}, \mathrm{T}\}.
\end{equation}

\begin{Definition}\label{def:1.1}
Let 
\(V=\mathbb{R}^{\lambda_1}\)
be the set of lists (vectors) whose components are indexed by the alphabet 
\(\mathcal{A} = \{\mathrm{A},\mathrm{C},\mathrm{G},\mathrm{T}\}\).

Any DNA sequence \(w\) of length \(L\) is represented by a frequency vector
\(x\in V\).
 Here, \(\mathrm{length}(x)\) denotes the number of components in the frequency vector \(x\) corresponding to the counts of each letter.

Let \(x_j \in \mathcal{A}\) for \(j = 1, \dots, \lambda_1\) be a letter of the alphabet. Every DNA sequence can be represented as a frequency vector 
\begin{equation}
x = (a_1,\, a_2,\, a_3,\, a_4),
\end{equation}
where \(a_j\) denotes the number of occurrences of the letter \(x_j\) in the sequence \(w\). By convention, we order the components such that
\begin{equation}
a_1 \ge a_2 \ge a_3 \ge a_4 \ge 0.
\end{equation}
\textbf{Note:} This ordering is adopted solely for computational efficiency; it loses the original correspondence between each count and its specific nucleotide (or \(n\)-tuple), which is irrelevant to the calculation of Shannon entropy. In our representation, it is also important to include components with zero count to avoid confounding later definitions.
\end{Definition}

\begin{Definition}\label{def:1.2}
From Definition~\ref{def:1.1}, we have that the frequency vector \(x\) has \(\lambda_1\) components, i.e.,
\begin{equation}
length(x) = \lambda_1
\end{equation}
\cite{ref19}
\end{Definition}

\begin{Definition}\label{def:1.3}
We can generalize the alphabet by grouping letters into sequences of length \(n\), called \emph{$n$-tuples}\cite{ref20}. The alphabet for \(n\)-tuples is denoted by \(\mathcal{A}_n\), where each element \(x_j \in \mathcal{A}_n\) is an \(n\)-tuple derived from \(\mathcal{A}\).
\end{Definition}

\begin{Definition}\label{def:1.4}
Since each position in an \(n\)-tuple can be any of the 4 letters, the size (cardinality) of the alphabet for \(n\)-tuples is
\begin{equation}
\lambda_n = 4^n
\end{equation}
Thus, by Definition~\ref{def:1.3}, the frequency vector \(x\) has 
\begin{equation}
length(x) = 4^n
\end{equation}
Accordingly, a sequence \(w\) can be represented by the vector
\begin{equation}
x = (a_1,\, a_2,\, \ldots,\, a_{\lambda_n}),
\end{equation}
where \(a_j\) denotes the occurrence count of the \(n\)-tuple \(x_j\). We conventionally order these counts in non-increasing order:
\begin{equation}
a_1 \ge a_2 \ge \cdots \ge a_{\lambda_n }\ge 0.
\end{equation}
(Compare with the ordering in Definition~\ref{def:1.1}.)

\textbf{Note:} For the sake of notational simplicity, we shall denote \(\lambda_n = 4^n\) simply as \(\lambda\).
Hence, throughout the paper we implicitly assume
\(\lambda=\lambda_n=4^n\); whenever \(n\) changes,
so does \(\lambda\).  In particular, for \(n=1\) we have
\(\lambda=4\).

\end{Definition}

\begin{Definition}\label{def:1.5}
\textbf{(Case \(n=1,\ L>0\))} Let 
\begin{equation}
b_j = \frac{a_j}{L}
\end{equation}
be the relative frequency of the letter \(x_j\) \emph{when \(n=1\)}, so that
\begin{equation}
\sum_{j=1}^\lambda a_j = L,
\end{equation}
and hence
\begin{equation}
\sum_{j=1}^\lambda b_j = 1.
\end{equation}

\textbf{(Case \(n>1,\ L\ge n\))} When \(n>1\) and the sequence \(w\) is divided into non-overlapping \(n\)-tuples, let 
\begin{equation}
M = \left\lfloor \frac{L}{n} \right\rfloor
\quad \text{and define}\quad
b_j = \frac{a_j}{M}.
\end{equation}
This is the normalized frequency of the letter or group of letters (n-tuple).
Hence 
\begin{equation}
\sum_{j=1}^\lambda a_j = M,
\quad\text{and}\quad
\sum_{j=1}^\lambda b_j = 1.
\end{equation}
\end{Definition}

\begin{Definition}\label{def:1.6}
A sequence \(w\) can be divided into ordered \(n\)-tuples, indexed by \(m\), indicating the position of the \(n\)-tuple in the sequence. By counting these \(n\)-tuples, we determine the occurrences \(a_j\) of each \(n\)-tuple type \(x_j\), which form the components of the vector \(x\).
\end{Definition}

\begin{Definition}\label{def:1.7.1}
Consider a sequence \(w = (w_1, w_2, \ldots, w_L)\) and let \(x_i^m\) denote the \(i\)-th letter in the \(m\)-th \(n\)-tuple. The set of \(n\)-tuples is said to be \emph{non-overlapping} if each letter in \(w\) is assigned to exactly one \(n\)-tuple. This is achieved by defining the correspondence
\begin{equation}
w_k \quad \text{with} \quad k = n(m-1) + i, \quad \text{for } i = 1,\dots,n, m = 1,\dots,M
\end{equation}
where \(M\) is the total number of non-overlapping \(n\)-tuples.
\end{Definition}

\begin{Definition}\label{def:1.8}
Here, \(\lfloor \cdot \rfloor\) denotes the integer part of the division. Note that, without such operation, \(M\) is an integer if and only if \(L\) is a multiple of \(n\). Otherwise, there will be a remainder of letters given by
\begin{equation}
z = L - n\,M = L \bmod n.
\end{equation}
\end{Definition}

\begin{Definition}\label{def:1.10}
The letters \(w_{L-z+1}\) to \(w_L\) (when \(z \neq 0\)) are not included in any \(n\)-tuple and are therefore ignored in subsequent calculations.

We can also divide the sequence \(w\) of length \(L\) into \(N\) equal subsequences\cite{ref21} \(w_t\), each of length \(T\), where 
\begin{equation}
N = \left\lfloor \frac{L}{T} \right\rfloor.
\end{equation}
\end{Definition}

\subsection{Shannon Entropy of a DNA Sequence \texorpdfstring{$w$}{w}}

To measure the entropy of a DNA sequence, one must first choose the method of assigning an entropy value\cite{ref22,ref23}. For example, one could refer to thermodynamic entropy\cite{ref23}, topological entropy\cite{ref6,ref17}, or Shannon entropy\cite{ref5}, among others. In this context, Shannon entropy will be used.

The classical Shannon entropy of a sequence \(w\), denoted \(S(w)\), is defined (based on previous definitions) as follows:
\begin{equation}
S(w) = -\sum_{j=1}^\lambda b_j \log(b_j),
\label{eq:2.1}
\end{equation}
where, by convention, \(\log\) denotes the base-2 logarithm. 

\begin{Definition}\label{def:2.1}
If \(b_j = 0\), we set \(b_j \log(b_j) = 0\) in the summation to represent the fact that this element is not counted in the calculation of Shannon entropy. This convention is useful for maintaining vectors with a consistent number of elements.
\end{Definition}

Here, Shannon entropy is essentially a measure of the informational complexity (in bits) of an alphabetical sequence\cite{ref6}.

Now, we address a previously overlooked question: why introduce a generalized alphabet of \(n\)-tuples?

The answer is straightforward: for very long DNA sequences generated by a uniform i.i.d. process, the distribution of letters becomes uniform, leading to
\begin{equation}
S(w)_{\text{with}\,L\gg 1} = -\sum_{j=1}^{4} b_j \log(b_j) \;\rightarrow\; \log(4) \;=\; 2,
\label{eq:2.2}
\end{equation}
since for \(L\gg 1\) (and \(n=1\)) \(b_j \simeq 0.25\) for all \(1 \le j \le 4\)\cite{ref7}.

Generally, for very long sequences and very short non-overlapping \(n\)-tuples generated by a uniform process:
\begin{equation}
S(w)_{\text{with}\,L\gg 1} = -\sum_{j=1}^\lambda b_j \log(b_j) \rightarrow \log(\lambda).
\label{eq:2.3}
\end{equation}

However, this issue becomes less pronounced for short \(n\)-tuples. Conversely, if we increase the length of the tuples, for example when \(n \to L\), the entropy \(S \to 0\) because very few \(n\)-tuples will be represented among the total possible \(\lambda\) \(n\)-tuples.

A further method, that distinguishes our approach, is the assignment of an \emph{average} entropy value\cite{ref24}  (which we will be referring to as \(S_w\) throughout the rest of the paper to avoid confusion with the classical case) to a sequence \(w\). This value is defined as the average of the entropies \(S^t_w\) of the \(t\)-th subsequences \(w_t\) of \(w\). This approach is computationally more efficient and avoids the aforementioned problem of saturating into the maximum possible entropy.
\\

In the following theorem, the index will be understood as that of a subsequence within a sequence.

\begin{Theorem}[Concatenation with exact bound]
\label{thm:1}
Fix $T,n$ and let $\lambda=4^n$. For sequences $o,v$ with lengths
$L_o,L_v$, write $L_o=N_oT+r_o$ and $L_v=N_vT+r_v$ with $0\le r_o,r_v<T$.
Let $S_o$ (resp.\ $S_v$) be the mean block entropy over the $N_o$ (resp.\ $N_v$) full $T$-blocks of $o$ (resp.\ $v$).
Form $w=o\|v$ and let $S_w$ be the mean over the first $N:=N_o+N_v$ full $T$-blocks of $w$.
Then
\begin{equation}
S_w=\frac{N_oS_o+N_vS_v}{N}.
\end{equation}
Moreover, with
\begin{equation}
\theta=\frac{L_oS_o+L_vS_v}{L_o+L_v},
\end{equation}
we have the bound
\begin{equation}
\bigl|\,\theta-S_w\,\bigr|\;\le\;\frac{\log(\lambda)}{TN}\,(r_o+r_v).
\end{equation}
In particular, when $r_o=r_v=0$ the identity $\theta=S_w$ holds exactly.
\end{Theorem}

\begin{proof}
Write $\theta=\beta S_o+(1-\beta)S_v$ with $\beta=\dfrac{L_o}{L_o+L_v}$ and
$S_w=\alpha S_o+(1-\alpha)S_v$ with $\alpha=\dfrac{N_o}{N}$, $N=N_o+N_v$.
Then
\begin{equation}
\theta-S_w=(\beta-\alpha)(S_o-S_v),
\end{equation}
so
\begin{equation}
\bigl|\,\theta-S_w\,\bigr|\le |\,\beta-\alpha\,|\cdot\bigl|S_o-S_v\bigr|\le |\,\beta-\alpha\,|\cdot\log(\lambda).
\end{equation}
Using $L_o=TN_o+r_o$ and $L_v=TN_v+r_v$ we compute
\begin{equation}
\beta-\alpha
=\frac{TN_o+r_o}{TN+r_o+r_v}-\frac{N_o}{N}
=\frac{r_oN_v-N_or_v}{N\bigl(TN+r_o+r_v\bigr)}.
\end{equation}
Hence
\begin{equation}
|\,\beta-\alpha\,|
\le \frac{r_oN_v+N_or_v}{N\bigl(TN+r_o+r_v\bigr)}
\le \frac{(r_o+r_v)N}{N\cdot TN}
=\frac{r_o+r_v}{TN}.
\end{equation}
Combining the inequalities gives the stated bound. If $r_o=r_v=0$, then $L_o=TN_o$, $L_v=TN_v$ and
\begin{equation}
\theta=\frac{TN_oS_o+TN_vS_v}{TN}=\frac{N_oS_o+N_vS_v}{N}=S_w.\qedhere
\end{equation}
\end{proof}

We generated three sets of entropy profiles to assess the behavior of the average Shannon entropy \(S_w\) under different slicing parameters.  In Figures \ref{fig:1.1.1} and \ref{fig:1.1.2}, \(S_w\) is shown as a function of the number \(N\) of non-overlapping subsequences (blocks) for a fixed sequence length \(L=1000\), using single bases (\(n=1\)) and triplets (\(n=3\)), respectively.  In both cases, the block size was set to 
\begin{equation}
T = \left\lfloor \frac{L}{N} \right\rfloor,
\end{equation}
and any residual bases were discarded.  Markers were placed every \(\Delta N=10\).  

\begin{figure}[H]
    \centering
    \includegraphics[width=\linewidth]{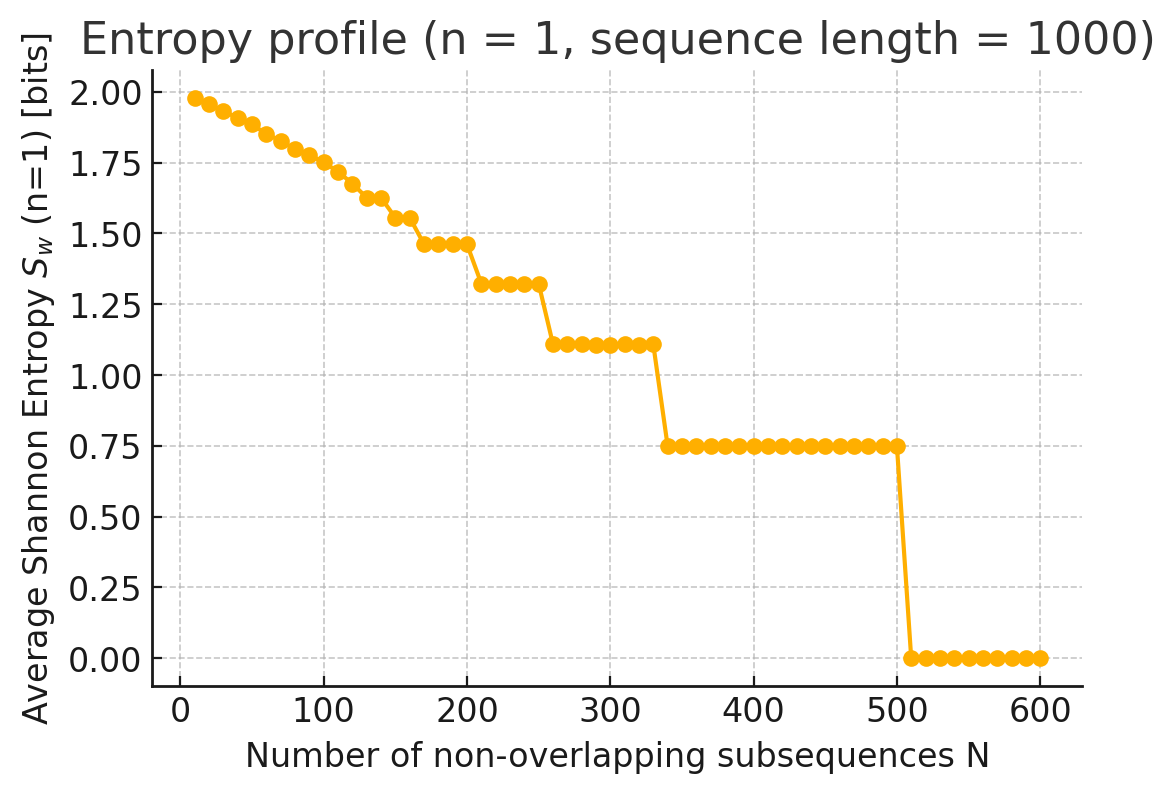}
    \caption{Average entropy \(S_w\) of a 1000-base random DNA sequence versus the number of non-overlapping subsequences \(N\), computed with single bases (\(n=1\)).  Each point is the mean over 50 independent random sequences; the decay is nonlinear and concave, and \(S_w\to0\) once \(T=1\) base.}
    \label{fig:1.1.1}
\end{figure}

\begin{figure}[H]
    \centering
    \includegraphics[width=\linewidth]{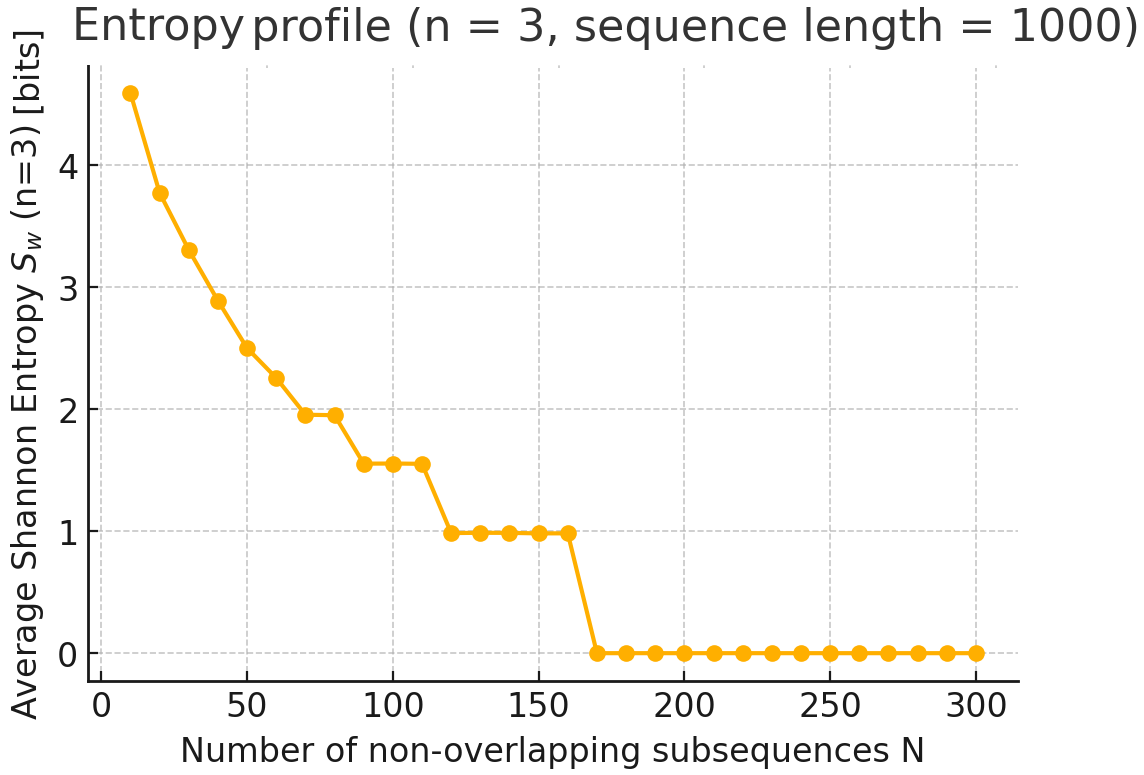}
    \caption{Average entropy \(S_w\) for the same 1000-base sequences, using triplets (\(n=3\)).  Points are averaged over 50 sequences; the entropy falls steeply and reaches zero when blocks no longer contain a full triplet (\(T<3\)).}
    \label{fig:1.1.2}
\end{figure}

In Figure \ref{fig:1.2}, we fix \(N=1\) (no blocking) and plot the Shannon entropy of the entire sequence as a function of the \(n\)-tuple length \(n\), for \(1\le n\le50\).  A clear maximum appears when the alphabet size \(4^n\) becomes comparable to the number of observed \(n\)-tuples; for \(L=1000\) this happens around \(n\approx6\). Beyond that point, undersampling dominates and entropy decreases as \(4^n\) grows relative to the sample size.

\begin{figure}[H]
    \centering
    \includegraphics[width=\linewidth]{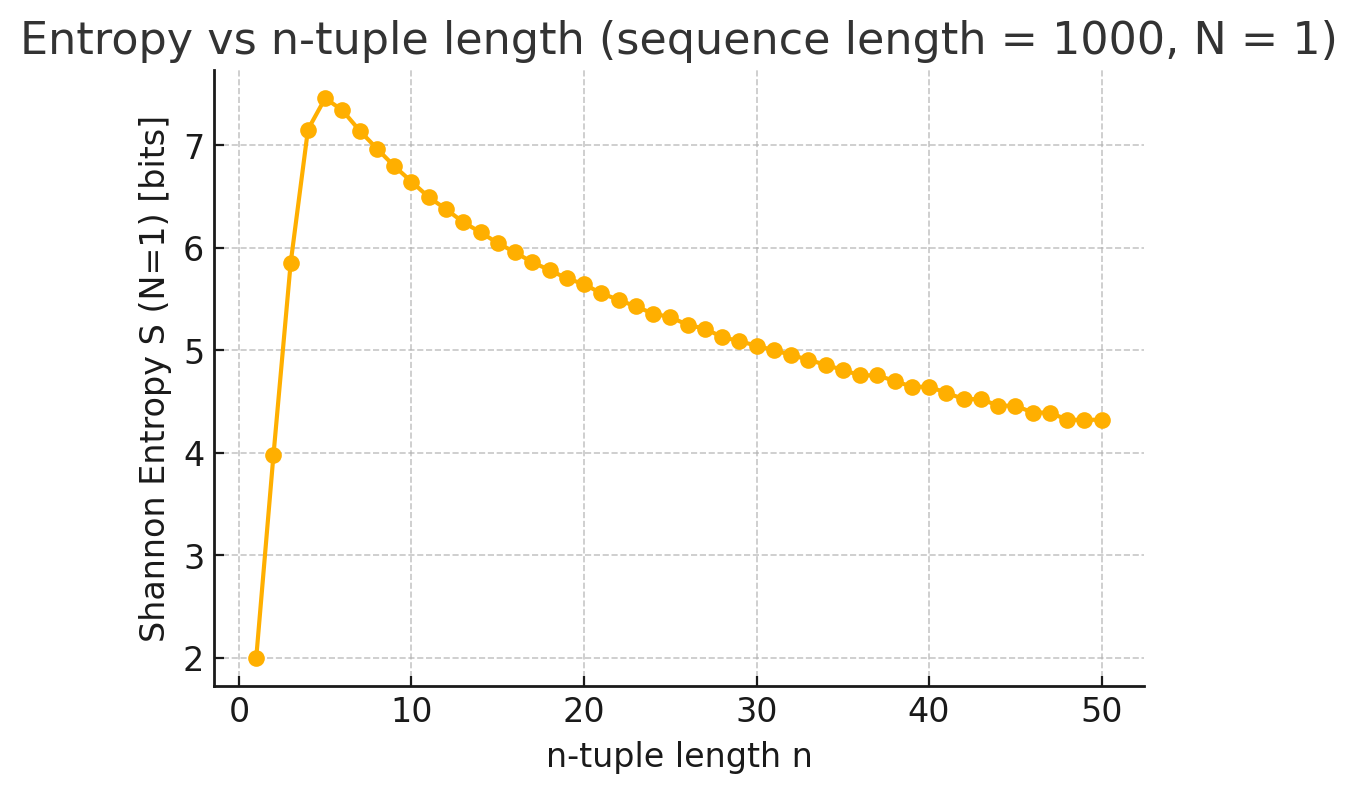}
    \caption{Shannon entropy \(S\) of a 1000-base random DNA sequence versus the \(n\)-tuple length \(n\) (no blocking, \(N=1\)), averaged over 50 sequences.  The curve peaks near \(n\approx6\) and then falls off as \(4^n\) approaches the sample size.}
    \label{fig:1.2}
\end{figure}

These profiles were obtained by Monte Carlo simulation over 50 independent random DNA sequences of length 1000.  For each sequence and each parameter setting (\(N\) or \(n\)), the sequence was partitioned into non-overlapping blocks of length \(T=\lfloor L/N\rfloor\)\,(or into \(n\)-tuples when \(N=1\)), Shannon entropy was computed in each block (or for the full sequence), and then averaged across blocks.  The resulting block-average entropies were finally averaged over all 50 sequences to produce the points shown.

It is important to note that the average entropy \(S_w\) of a sequence \(w\) is \emph{relative} to the chosen values for the \(n\)-tuple length \(n\) and the subsequence length \(T\). The selection of these parameters must allow the assignment of a meaningful average entropy value. Since the concept of entropy is inherently relative, the numerical value of our average entropy (a real scalar) does not have an absolute interpretation; its meaning depends entirely on the reference system (i.e., the chosen parameters). This is the main limitation of the entropy calculation system based on Shannon entropy (aside from the saturation issues described in Equation~\eqref{eq:2.3}). 

\subsection{Giving Meaning to the Entropy Value}

Many previous works have attempted to use alternative entropy systems to obtain distributions or a set of values in order to study the complexity of a sequence\cite{ref25,ref26,ref27}. Here, we introduce the concept of the \emph{distribution of all possible sequences of the same length corresponding to a single entropy value}.

To explain the underlying intuition, consider that a DNA sequence \(w\) of length \(L\) can be one among many sequences of the same length but with different arrangements of letters. For example, take the sequences \texttt{AAAG} (denoted \(o\)) and \texttt{GGTT} (denoted \(v\)); both have \(L=4\) but different letter compositions. If we group their letters into frequency vectors with \(n=1\),
\begin{equation}
x_o = \bigl(\; 3_\mathrm{A},1_\mathrm{G}\bigr) \quad\text{and}\quad
x_v = \bigl(2_\mathrm{G},\; 2_\mathrm{T}\bigr),
\end{equation}
it is evident that these sequences will, in general, exhibit slightly different entropy values (approximately \(0.811278\) and \(1\), respectively). The difference between the entropies of \(o\) and \(v\) is relative to the parameters chosen; for instance, when \(T=4\) (so that the entire sequence constitutes a single subsequence) and \(n=1\), sequence \(o\) yields a lower entropy than \(v\). Such differences are relative and do not possess an absolute meaning when considered in isolation.

Nonetheless, a more meaningful interpretation of the entropy of a sequence can be achieved by relating it not only to another sequence but to \emph{all possible sequences of length \(L\)}. In this way, one derives a measure of the intrinsic complexity of a single sequence, independent. To draw an analogy, imagine two individuals with heights of 1.63 meters and 1.89 meters. Without any reference, one may simply say one is shorter than the other, yet absolute labels such as “short” or “tall” only acquire meaning when compared to the average height of a population. Similarly, rather than comparing individual entropy values to a mean, one may count how many DNA subsequences (with fixed length \(T\) and non-overlapping \(n\)-tuple grouping) exhibit an entropy lower than that of a given sequence \(w\) (computed using the same values of \(T\) and \(n\)). This method yields a distribution of all possible sequences for a given \(T\), ordered by entropy ranging from \(0\) to \(\log(\lambda)\), thus enabling the placement of \(w\) within the overall spectrum of complexity.

\subsection{Formalizing the Concept}

Fix a block (subsequence) length \(T\) and an \(n\)-tuple size \(n\), with non-overlapping \(n\)-tuples inside each block. Let
\begin{equation}
M=\left\lfloor \frac{T}{n}\right\rfloor
\end{equation}
be the number of \(n\)-tuples per block (for \(n=1\) we have \(M=T\)). The effective alphabet size is \(\lambda=\lambda_n=4^n\). Throughout this section we work at the \emph{block level} (i.e., on a segment of length \(T\)).

\begin{Definition}\label{def:2.2}
Let \(\aleph_{T,n}\) denote the total number of distinct entropy values obtainable on a block of length \(T\) with \(n\)-tuples (non-overlapping). Define
\begin{equation}
Y_{T,n}=\{\,y_1,\;y_2,\;\ldots,\;y_{\aleph_{T,n}}\,\}
\end{equation}
as the set of these entropy values, ordered so that
\begin{equation}
\log(\lambda) \ge y_{\aleph_{T,n}} \ge \cdots \ge y_2 \ge y_1=0.
\end{equation}
\end{Definition}

\begin{Definition}\label{def:2.3}
Let 
\begin{equation}
G_{T,n}(y_q) : Y_{T,n} \to \mathbb{N}
\end{equation}
assign to each \(y_q\in Y_{T,n}\) the number of words (blocks) attaining that entropy value.
\end{Definition}

\begin{Definition}\label{def:2.4}
The total number of possible blocks equals
\begin{equation}
\sum_{q=1}^{\aleph_{T,n}} G_{T,n}(y_q) \;=\; \lambda^{\,M} \;=\; (4^n)^{\,\lfloor T/n\rfloor}.
\end{equation}
If \(T \bmod n \neq 0\), the last \(z=T\bmod n\) symbols of the block are ignored for the entropy computation (truncation); the counting always refers to the truncated word of length \(nM\).
\end{Definition}

An important insight follows by representing frequency vectors via integer partitions \cite{ref29}. Let
\begin{equation}
P = (p_1,\ldots,p_\lambda)
\end{equation}
be an \emph{ordered} partition of the total count \(c\) (with \(c=T\) for \(n=1\), or \(c=M\) for \(n>1\)):
\begin{equation}
p_1 \ge \cdots \ge p_\lambda \ge 0,\qquad \sum_{i=1}^{\lambda} p_i = c.
\end{equation}
Denote by \(\mathcal{P}_{\lambda}(c)\) the set of such ordered partitions.

\begin{Definition}\label{def:2.6}
For \(P\in\mathcal{P}_{\lambda}(c)\) define the Shannon entropy
\begin{equation}
S(P) = -\sum_{i=1}^{\lambda} \frac{p_i}{c}\log\!\left(\frac{p_i}{c}\right),
\end{equation}
with the convention \(0\log 0=0\). This induces a mapping
\(
S:\mathcal{P}_{\lambda}(c)\to\mathbb{R}.
\)
\end{Definition}

Let \(k=\#\{i:\,p_i>0\}\) be the number of positive parts. Let \(\{\pi_j\}\) be the distinct positive values among \(\{p_i\}\) and \(r_j\) their multiplicities; set \(\gamma=\prod_j r_j!\). The number of words realizing the frequency vector \(P\) is
\begin{equation}
O(P)
= \binom{\lambda}{k}\,\frac{k!}{\gamma}\;\cdot\;\frac{c!}{\prod_{i=1}^{\lambda} p_i!}
\;=\; \frac{\lambda!}{(\lambda-k)!\,\gamma}\;\frac{c!}{\prod_{i=1}^{\lambda} p_i!}.
\end{equation}

\begin{proposition}[Sanity check]
\label{prop:sumOP}
We have \(\displaystyle \sum_{P\in\mathcal{P}_{\lambda}(c)} O(P) = \lambda^{\,c}\).
\end{proposition}

\begin{proof}
Let \(\Sigma\) be an alphabet of size \(\lambda\) and fix a length \(c\in\mathbb{N}\).
For a word \(w\in\Sigma^{c}\) let \(f(w)=(f_1,\dots,f_\lambda)\in\mathbb{N}^\lambda\) be its
frequency vector, so that \(\sum_{i=1}^\lambda f_i=c\).
For a fixed frequency vector \(f\), the number of words that realize exactly \(f\) is
\begin{equation}
\frac{c!}{\prod_{i=1}^\lambda f_i!}.
\end{equation}

Group frequency vectors \(f\) by their \emph{ordered partition} \(P\) obtained by sorting
the components of \(f\) in nonincreasing order. Write
\(P=(p_1,\dots,p_\lambda)\) with \(p_1\ge\cdots\ge p_\lambda\ge 0\),
\(\sum_i p_i=c\), and let \(k=\#\{i:\,p_i>0\}\).
Let \(\{\pi_j\}\) be the distinct positive values among \(\{p_i\}\) and
let \(r_j\) be their multiplicities; set \(\gamma=\prod_j r_j!\).

For a fixed \(P\), to recover the unsorted frequency vectors \(f\) that sort to \(P\):
(i) choose which \(k\) symbols of \(\Sigma\) are assigned the positive counts
(\(\binom{\lambda}{k}\) ways); (ii) assign the multiset of positive parts
\(\{\pi_j\}\) to those \(k\) symbols (\(k!/\gamma\) ways, accounting for repeated
values). Hence the number of frequency vectors \(f\) that sort to \(P\) is
\(\binom{\lambda}{k}\,k!/\gamma=\lambda!/((\lambda-k)!\,\gamma)\).
Each such \(f\) contributes exactly \(c!/\prod_i p_i!\) words (since \(p_i\) are
the components of \(f\) in some order and \(0!=1\)).
Therefore the total number of words with “shape” \(P\) equals
\begin{equation}
O(P)=\frac{\lambda!}{(\lambda-k)!\,\gamma}\;\frac{c!}{\prod_{i=1}^{\lambda} p_i!}.
\end{equation}

Summing over all ordered partitions \(P\in\mathcal{P}_{\lambda}(c)\) thus yields
\begin{equation}
\sum_{P\in\mathcal{P}_{\lambda}(c)} O(P)
=\sum_{\substack{f_1,\dots,f_\lambda\ge 0\\ f_1+\cdots+f_\lambda=c}}
\frac{c!}{\prod_{i=1}^\lambda f_i!}
=(1+\cdots+1)^c=\lambda^{\,c}.
\end{equation}
The last equality is the multinomial theorem applied to
\((x_1+\cdots+x_\lambda)^c\) with \(x_i\equiv 1\).
\end{proof}

Since different partitions can share the same entropy, define the discrete distribution
\begin{equation}
G_{T,n}(y) \;=\; \sum_{P\,:\,S(P)=y} O(P).
\end{equation}
This is the distribution used in Definitions~\ref{def:2.2}--\ref{def:2.4}.

We can now construct \(G_{T,n}(y)\). For example, Figure \ref{fig:2} shows the case \(T=20\) and \(n=1\), where \(c=T\) and \(\lambda=4\).

\begin{figure}[H]
    \centering
    \includegraphics[width=1.00\textwidth]{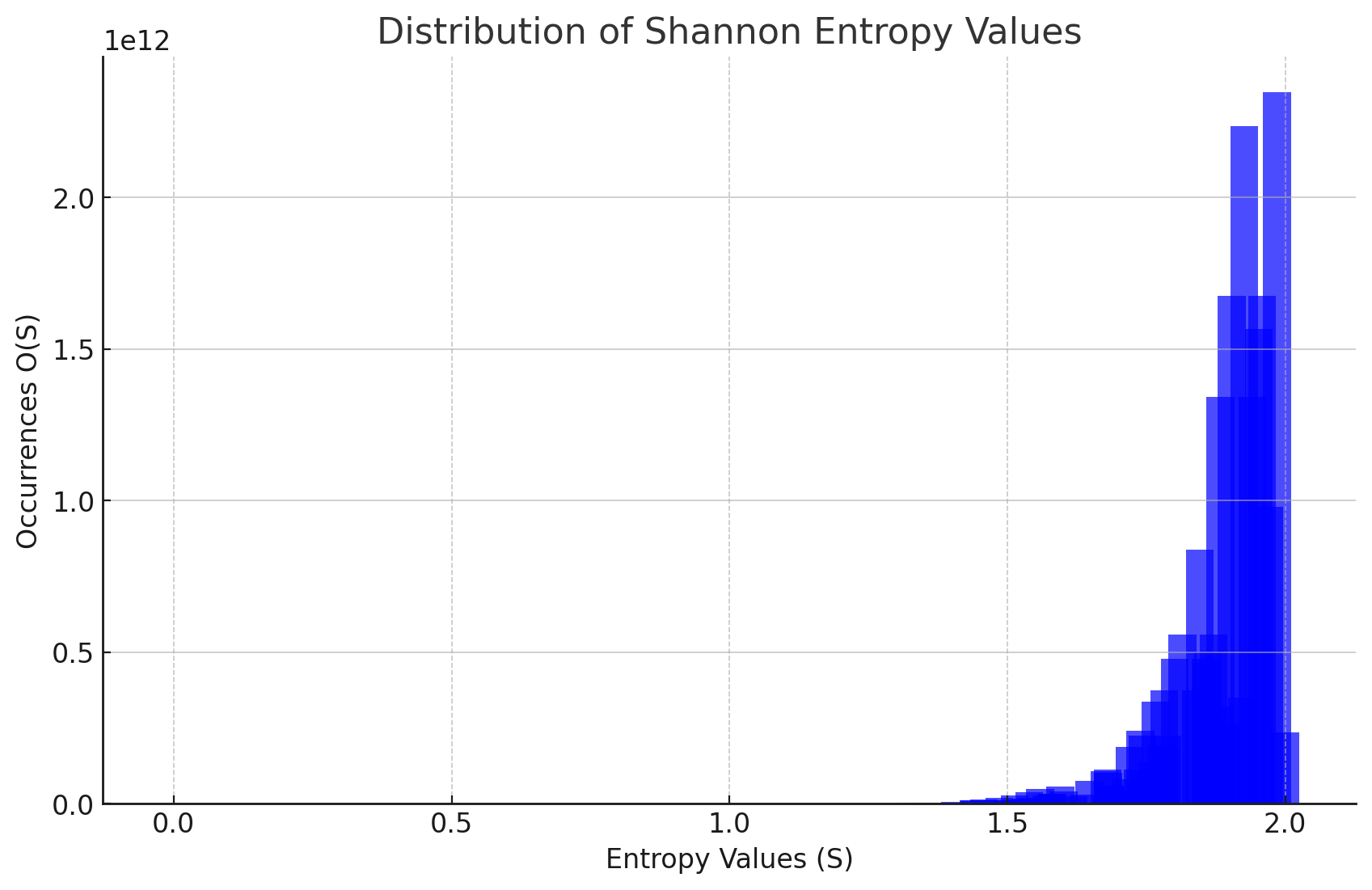}
    \caption{Distribution \(G\) of occurrences \(O\) as a function of entropy values \(S\) (plotted on the horizontal axis, arranged in increasing order from 0 to 2.0), using \(T=20\) and \(n=1\). As described earlier, \(G\) is a discrete distribution.}
    \label{fig:2}
\end{figure}
\begin{Definition}[Entropy–rank ratio]\label{def:R}
Fix the subsequence length \(T\), the \(n\)-tuple size \(n\), and the number of blocks \(N\).
Let \(S_{T,n}^{(N)}(w)\) be the mean Shannon entropy of \(w\) computed by
splitting \(w\) into \(N\) consecutive blocks of length \(T\) and averaging
their block entropies (with non-overlapping \(n\)-tuples inside each block).
Let \(G_{T,n}^{(N)}\) be the distribution of the mean of \(N\) i.i.d. block entropies, i.e. the \(N\)-fold discrete convolution of \(G_{T,n}\) for the sum, followed by a scaling of the support by \(1/N\).

The \emph{entropy–rank ratio} of \(w\) is
\begin{equation}
R_{T,n}^{(N)}(w)\;=\;
\frac{\displaystyle
      \sum_{y \,\leq\, S_{T,n}^{(N)}(w)} G_{T,n}^{(N)}(y)}
     {\displaystyle
      \sum_{y} G_{T,n}^{(N)}(y)}
\;\in (0,1].
\end{equation}

\textit{Practical note.} In our implementation we use \(N=1\), so \(S_{T,n}^{(1)}(w)\)
is the entropy of a single block of length \(T\) and \(G_{T,n}^{(1)}\equiv G_{T,n}\).

\end{Definition}

\begin{proposition}[Block-mean counting as $N$-fold convolution]
\label{prop:Nfold}
Fix $T,n$ and let $G_{T,n}(y)$ be the block-level counting measure of entropies on length-$T$ blocks. Consider sequences of length $NT$ segmented in $N$ consecutive $T$-blocks. Let $G_{T,n}^{(N)}(s)$ count how many such sequences have mean block entropy exactly $s\in \frac{1}{N}\sum_{i=1}^N \mathrm{supp}(G_{T,n})$. Then
\begin{equation}
G_{T,n}^{(N)}(s)
=\sum_{\substack{y_1,\dots,y_N\\(y_1+\cdots+y_N)/N=s}}
\;\prod_{i=1}^N G_{T,n}(y_i),
\end{equation}
i.e.\ $G_{T,n}^{(N)}$ is the $N$-fold discrete convolution of $G_{T,n}$ for the sum, followed by a scaling of the support by $1/N$. Moreover $\sum_s G_{T,n}^{(N)}(s) = \bigl(\sum_y G_{T,n}(y)\bigr)^N$.
\end{proposition}

\begin{proof}
Every length-$NT$ sequence is an ordered concatenation of $N$ length-$T$ blocks. The number of concatenations realizing a fixed $N$-tuple $(y_1,\dots,y_N)$ is $\prod_i G_{T,n}(y_i)$. Grouping by the mean $s=(y_1+\cdots+y_N)/N$ yields the formula. The total count follows by the multiplicative principle.
\end{proof}

\begin{proposition}[Calibration of $R$]
\label{prop:calibration}
Let $Y$ be a random entropy drawn from the counting measure $G_{T,n}$ on length-$T$ blocks (uniform over all blocks after truncation). Define the right-continuous cdf
\begin{equation}
F(y)=\frac{\sum_{u\le y} G_{T,n}(u)}{\sum_{u} G_{T,n}(u)}.
\end{equation}
Then $R=F(Y)$ is super-uniform: for all $t\in[0,1]$,
\begin{equation}
\mathbb{P}(R\le t)\;\le\; t,
\end{equation}
with equality if ties at atoms of $F$ are broken uniformly at random. The same statement holds for $N$ blocks with $G_{T,n}^{(N)}$.
\end{proposition}

\begin{proof}
This is the standard property of the cdf in the discrete case: $F(Y)$ is right-continuous uniform (or super-uniform without randomized tie-breaking). The $N$-block case follows by replacing $G_{T,n}$ with $G_{T,n}^{(N)}$.
\end{proof}

\subsection{Implementation}

In our implementation, the DNA sequences are first converted into integer arrays by mapping each nucleotide \(A\), \(C\), \(G\), \(T\) to 0, 1, 2, 3, respectively for tokenization, while reserving the value 4 for the padding symbol. For the purpose of entropy computation, each sequence is partitioned into non‑overlapping \(n\)-tuples and the frequency vector is computed from these groups. The Shannon entropy is then calculated based on Equation~\eqref{eq:2.1}. Data augmentation is achieved by cropping the sequence into fixed‑length segments using either:\\
1. A Random Crop.\\
2. A Kolmogorov complexity‑based crop, that selects subsequences whose Kolmogorov complexity, calculated with the help of the zlib library, is close to that of the entire sequence.\\
3. An entropy‑based crop that selects subsequences whose Shannon entropy is close to that of the entire sequence.\\
4. A Ratio‑based crop that selects subsequences whose entropy ratio \(R\) is close to that of the entire sequence. \\
5. Basic Crop that provides technically no real augmentation which will be the reference (we are not comparing against full‐sequence models, as our sole objective is to evaluate the efficiency of the complexity measures).\\
The CNN classifier uses an embedding layer with a vocabulary size of 5 (corresponding to the four nucleotides plus the padding token), followed by convolutional layers, adaptive pooling, dropout, and a fully connected layer for classification. Although many CNN‑based approaches utilize k‑mer tokenization\cite{ref33} with larger vocabularies, our implementation adheres to single‑nucleotide tokenization combined with non‑overlapping \(n\)-tuple based entropy analysis.

\begin{algorithm}[H]
\caption{Random Crop}
\begin{algorithmic}[1]
\Procedure{RandomCrop}{seq, target\_len, offset\_ratio}
    \State \textbf{input:} seq, target\_len, offset\_ratio
    \State \textbf{if} $\mathrm{length}(\mathrm{seq}) \leq \mathrm{target\_len}$ \textbf{then}
        \State \quad \textbf{return} \texttt{pad\_or\_sample\_sequence}(seq, $\mathrm{target\_len}$)
    \State \textbf{else}
        \State \quad $\mathrm{center} \gets \left\lfloor \dfrac{\mathrm{length}(\mathrm{seq}) - \mathrm{target\_len}}{2} \right\rfloor$
        \State \quad $\mathrm{max\_offset} \gets \left\lfloor \bigl(\mathrm{length}(\mathrm{seq}) - \mathrm{target\_len}\bigr)\times \mathrm{offset\_ratio} \right\rfloor$
        \State \quad $\mathrm{offset} \gets \mathrm{RandomInteger}(-\mathrm{max\_offset},\, \mathrm{max\_offset})$
        \State \quad $\mathrm{start} \gets \mathrm{Clip}\bigl(\mathrm{center} + \mathrm{offset},\, 0,\, \mathrm{length}(\mathrm{seq}) - \mathrm{target\_len}\bigr)$
    \State \textbf{endif}
    \State \textbf{return} seq[$\mathrm{start} : \mathrm{start} + \mathrm{target\_len}$]
\EndProcedure
\end{algorithmic}
\end{algorithm}

\noindent
\textbf{Algorithm: Random Crop.}\\
This procedure crops a sequence to a fixed length while allowing controlled randomness around its center. If the sequence is shorter than or equal to the target length, it is processed by \texttt{pad\_or\_sample\_sequence}. Otherwise, the center of the sequence is computed, a maximum offset is determined by the given \texttt{offset\_ratio}, and a random offset is applied. The starting index is then clipped to ensure it is within the valid range, and the fixed-length subsequence is returned.

\begin{algorithm}[H]
\caption{R-based Crop}
\begin{algorithmic}[1]
\Procedure{RatioBasedCrop}{seq, target\_len, num\_candidates, offset\_ratio, ratio\_whole\_seq, $T$, $n$, $\alpha$, $\beta$}
    \State \textbf{input:} seq, target\_len, num\_candidates, offset\_ratio, ratio\_whole\_seq, $T$, $n$, $\alpha$, $\beta$
    \State \textbf{if} $\mathrm{length}(\mathrm{seq}) \leq \mathrm{target\_len}$ \textbf{then}
        \State \quad \textbf{return} \texttt{pad\_or\_sample\_sequence}(seq, $\mathrm{target\_len}$)
    \State \textbf{endif}
    \State arr $\gets$ \texttt{Array}(seq) \Comment{0-based indexing}
    \State $L \gets \mathrm{length}(\mathrm{arr})$
    \State $\mathrm{center} \gets \left\lfloor \dfrac{L - \mathrm{target\_len}}{2} \right\rfloor$
    \State $\mathrm{max\_offset} \gets \left\lfloor (L - \mathrm{target\_len})\times \mathrm{offset\_ratio} \right\rfloor$
    \State candidate\_offsets $\gets$ $\texttt{RandomIntegers}(-\mathrm{max\_offset},\, \mathrm{max\_offset},\, \mathrm{num\_candidates})$
    \State $\mathrm{best\_score} \gets \infty$; \quad $\mathrm{best\_candidate} \gets \mathrm{None}$
    \State \textbf{for} $i \gets 0$ \textbf{to} $\mathrm{num\_candidates}-1$ \textbf{do}
        \State \quad $\mathrm{start} \gets \mathrm{center} + \mathrm{candidate\_offsets}[i]$
        \State \quad \textbf{if} $\mathrm{start} < 0$ \textbf{then} $\mathrm{start} \gets 0$ \textbf{endif}
        \State \quad \textbf{if} $\mathrm{start} > L - \mathrm{target\_len}$ \textbf{then} $\mathrm{start} \gets L - \mathrm{target\_len}$ \textbf{endif}
        \State \quad segment $\gets$ arr[$\mathrm{start} : \mathrm{start} + \mathrm{target\_len}$]
        \State \quad $\mathrm{cand\_ratio} \gets \texttt{calculate\_ratio}(\mathrm{segment},\, T,\, n)$ \Comment{same $(T,n)$ as \texttt{ratio\_whole\_seq}}
        \State \quad $\mathrm{ratio\_diff} \gets \bigl|\,\mathrm{cand\_ratio} - \mathrm{ratio\_whole\_seq}\,\bigr|$
        \State \quad \textbf{if} $\mathrm{max\_offset} > 0$ \textbf{then}
            \State \qquad $\mathrm{offset\_penalty} \gets \dfrac{\bigl|\,\mathrm{candidate\_offsets}[i]\,\bigr|}{\mathrm{max\_offset}}$
        \State \quad \textbf{else}
            \State \qquad $\mathrm{offset\_penalty} \gets 0$
        \State \quad \textbf{endif}
        \State \quad $\mathrm{score} \gets \alpha \times \mathrm{ratio\_diff} + \beta \times \mathrm{offset\_penalty}$
        \State \quad \textbf{if} $\mathrm{score} < \mathrm{best\_score}$ \textbf{then}
            \State \qquad $\mathrm{best\_score} \gets \mathrm{score}$; \quad $\mathrm{best\_candidate} \gets \mathrm{segment}$
        \State \quad \textbf{endif}
    \State \textbf{endfor}
    \State \textbf{return} $\mathrm{best\_candidate}$
\EndProcedure
\end{algorithmic}
\end{algorithm}

\noindent
\textbf{Algorithm: R-based Crop.}\\
This procedure crops a sequence to a fixed length using a ratio-based selection mechanism. If the sequence is shorter than or equal to the target length, it is processed by \texttt{pad\_or\_sample\_sequence}. Otherwise, the sequence is converted into an array, a central index is computed, and a set of candidate offsets is generated. For each candidate, the starting index is clipped, and a subsequence is extracted. A ratio for the candidate segment is computed and compared to the provided \texttt{ratio\_whole\_seq}. A score is calculated by combining the ratio difference (weighted by $\alpha$) and an offset penalty (weighted by $\beta$). The candidate with the minimum score is selected and returned as the best representative subsequence.

\textbf{Note:}
Parameters $\alpha$ and $\beta$ weigh, respectively, ratio coherence and offset penalty; they must be tuned on validation (e.g., grid/random search). Setting $\beta=0$ removes recentering pressure.

\begin{algorithm}[H]
\caption{Compress Subchunk}
\begin{algorithmic}[1]
\Procedure{compressSubchunk}{chunk}
    \State \textbf{if} $\mathrm{size}(\mathrm{chunk}) = 0$ \textbf{then}
        \State \quad \textbf{return} $(0, \mathrm{chunk})$
    \State \textbf{endif}
    \State bytes\_buf $\gets$ \texttt{PackAsBytes}(chunk) \Comment{e.g. cast to uint8 and pack values 0--4}
    \State cached\_val $\gets$ Retrieve \texttt{bytes\_buf} from \texttt{\_KOLMOGOROV\_CACHE}
    \State \textbf{if} $\mathrm{cached\_val}$ is not $\mathrm{None}$ \textbf{then}
        \State \quad \textbf{return} $(\mathrm{cached\_val},\, \mathrm{chunk})$
    \State \textbf{endif}
    \State comp $\gets$ \texttt{Compress}(bytes\_buf, level $=$ 1)
    \State $\mathrm{length} \gets \mathrm{Length}(\mathrm{comp})$
    \State \textbf{if} $\mathrm{Size}(\texttt{\_KOLMOGOROV\_CACHE}) < \texttt{\_KOLMOGOROV\_CACHE\_MAX\_SIZE}$ \textbf{then}
        \State \quad Store (\texttt{bytes\_buf}, $\mathrm{length}$) in \texttt{\_KOLMOGOROV\_CACHE}
    \State \textbf{endif}
    \State \textbf{return} $(\mathrm{length},\, \mathrm{chunk})$
\EndProcedure
\end{algorithmic}
\end{algorithm}

\noindent
\textbf{Algorithm: Compress Subchunk.}\\
This procedure compresses a given subchunk. If the subchunk is empty, it returns zero. Otherwise, it packs the integers into a binary buffer (e.g., uint8), checks a cache for a precomputed compressed length, compresses the buffer if needed, caches the result (subject to a maximum size), and returns the compressed length along with the original subchunk.

\vspace{1em}

\begin{algorithm}[H]
\caption{Kolmogorov-based Crop}
\begin{algorithmic}[1]
\Procedure{kolmogorovBasedCrop}{arr, target\_len, offset\_ratio = \textit{OFFSET\_RATIO},
                                pick = \texttt{'max'}, num\_candidates = \textit{NUM\_CANDIDATES}}
    \State $L \gets \mathrm{length}(\mathrm{arr})$
    \State \textbf{if} $L \leq \mathrm{target\_len}$ \textbf{then}
        \State \quad \textbf{return} arr
    \State \textbf{endif}
    \State $\mathrm{center} \gets (L - \mathrm{target\_len}) // 2$
    \State $\mathrm{max\_offset} \gets \left\lfloor (L - \mathrm{target\_len})\times \mathrm{offset\_ratio} \right\rfloor$
    \State tasks $\gets$ empty list
    \State \textbf{for} $j \gets 1$ \textbf{to} $\mathrm{num\_candidates}$ \textbf{do}
        \State \quad $\mathrm{off} \gets \texttt{randint}(-\mathrm{max\_offset},\, \mathrm{max\_offset})$
        \State \quad $\mathrm{start} \gets \mathrm{center} + \mathrm{off}$
        \State \quad \textbf{if} $\mathrm{start} < 0$ \textbf{then}
            \State \qquad $\mathrm{start} \gets 0$
        \State \quad \textbf{elseif} $\mathrm{start} > L - \mathrm{target\_len}$ \textbf{then}
            \State \qquad $\mathrm{start} \gets L - \mathrm{target\_len}$
        \State \quad \textbf{endif}
        \State \quad chunk $\gets$ arr[$\mathrm{start} : \mathrm{start} + \mathrm{target\_len}$]
        \State \quad Append $\mathrm{chunk}$ to tasks
    \State \textbf{endfor}
    \State \textbf{if} \texttt{pick} $=$ \texttt{'max'} \textbf{then}
        \State \quad $\mathrm{best\_val} \gets -\infty$
    \State \textbf{else}
        \State \quad $\mathrm{best\_val} \gets \infty$
    \State \textbf{endif}
    \State $\mathrm{best\_chunk} \gets \mathrm{None}$
    \State \textbf{with} \texttt{ThreadPoolExecutor}$(\mathrm{max\_workers} = \textit{NUM\_KOLMOGOROV\_THREADS})$ \textbf{as} ex
        \State \quad \textbf{for each} $(\mathrm{comp\_len},\, \mathrm{chunk})$ \textbf{in} ex.\texttt{map}(\textsc{compressSubchunk}, tasks)
            \State \qquad \textbf{if} $\bigl(\texttt{pick} = \texttt{'max'} \textbf{ and } \mathrm{comp\_len} > \mathrm{best\_val}\bigr)$ \textbf{or}
            \State \qquad\qquad $\bigl(\texttt{pick} \neq \texttt{'max'} \textbf{ and } \mathrm{comp\_len} < \mathrm{best\_val}\bigr)$ \textbf{then}
                \State \qquad\quad $\mathrm{best\_val} \gets \mathrm{comp\_len}$
                \State \qquad\quad $\mathrm{best\_chunk} \gets \mathrm{chunk}$
            \State \qquad \textbf{endif}
    \State \textbf{return} $\mathrm{best\_chunk}$ \textbf{if} $\mathrm{best\_chunk} \neq \mathrm{None}$ \textbf{else} arr
\EndProcedure
\end{algorithmic}
\end{algorithm}

\noindent
\textbf{Algorithm: Kolmogorov-based Crop.}\\
This procedure crops an array to a fixed length using a Kolmogorov complexity–based criterion. If the array length is less than or equal to the target length, it returns the original array. Otherwise, it generates multiple candidate subarrays by applying random offsets around the center. For each candidate, it computes a compressed length using the \textsc{compressSubchunk} procedure. Depending on whether the selection criterion is 'max' or otherwise, it selects the candidate with the maximum or minimum compressed length and returns it.

\vspace{1em}

\begin{algorithm}[H]
\caption{Entropy-based Crop}
\begin{algorithmic}[1]
\Procedure{EntropyBasedCrop}{arr, target\_len, full\_entropy, $n$, num\_candidates, offset\_ratio, $\alpha$, $\beta$}
    \State $L \gets \mathrm{length}(\mathrm{arr})$
    \State \textbf{if} $L \leq \mathrm{target\_len}$ \textbf{then}
        \State \quad \textbf{return} arr
    \State \textbf{endif}
    \State $\mathrm{center} \gets \left\lfloor \dfrac{L - \mathrm{target\_len}}{2} \right\rfloor$
    \State $\mathrm{max\_offset} \gets \left\lfloor (L - \mathrm{target\_len})\times \mathrm{offset\_ratio} \right\rfloor$
    \State \textbf{allocate} arrays $\mathrm{offsets}[0..\mathrm{num\_candidates}-1]$, $\mathrm{scores}[0..\mathrm{num\_candidates}-1]$, $\mathrm{starts}[0..\mathrm{num\_candidates}-1]$
    \State \textbf{for} $i=0$ \textbf{to} $\mathrm{num\_candidates}-1$ \textbf{do}
        \State \quad $\mathrm{offsets}[i] \gets \mathrm{RandomInteger}(-\mathrm{max\_offset},\, \mathrm{max\_offset})$
    \State \textbf{endfor}
    \State \textbf{for} $i=0$ \textbf{to} $\mathrm{num\_candidates}-1$ \textbf{do}
        \State \quad $\mathrm{off} \gets \mathrm{offsets}[i]$
        \State \quad $\mathrm{start} \gets \mathrm{Clip}(\mathrm{center} + \mathrm{off},\, 0,\, L - \mathrm{target\_len})$
        \State \quad $\mathrm{starts}[i] \gets \mathrm{start}$
        \State \quad chunk $\gets$ arr[$\mathrm{start} : \mathrm{start} + \mathrm{target\_len}$]
        \State \quad $\mathrm{cand\_ent} \gets \texttt{compute\_n\_gram\_entropy\_jit}(\mathrm{chunk},\, n)$ \Comment{entropy of the chunk, $N=1$}
        \State \quad $\mathrm{ent\_diff} \gets \bigl|\,\mathrm{cand\_ent} - \mathrm{full\_entropy}\,\bigr|$
        \State \quad \textbf{if} $\mathrm{max\_offset} > 0$ \textbf{then}
            \State \qquad $\mathrm{offset\_penalty} \gets |\mathrm{off}|/\mathrm{max\_offset}$
        \State \quad \textbf{else}
            \State \qquad $\mathrm{offset\_penalty} \gets 0$
        \State \quad \textbf{endif}
        \State \quad $\mathrm{scores}[i] \gets \alpha \times \mathrm{ent\_diff} + \beta \times \mathrm{offset\_penalty}$
    \State \textbf{endfor}
    \State $\mathrm{best\_idx} \gets \arg\min_i \,\mathrm{scores}[i]$
    \State \textbf{return} arr[$\mathrm{starts}[\mathrm{best\_idx}] : \mathrm{starts}[\mathrm{best\_idx}] + \mathrm{target\_len}$]
\EndProcedure
\end{algorithmic}
\end{algorithm}

\noindent
\textbf{Algorithm: Entropy-based Crop.}\\
This procedure crops an array to a fixed length by selecting the segment whose n‑gram entropy best approximates a given full entropy value. If the array is too short, it returns the original array. Otherwise, it generates multiple candidate segments using random offsets around the center. For each candidate, it computes the entropy difference and applies an offset penalty. The candidate with the minimal combined score (weighted by $\alpha$ and $\beta$, two real numbers) is chosen and returned. Any \(n\)-gram containing the padding token is ignored in the entropy calculation.

\vspace{1em}

All the previous algorithms are precomputed augmentations that are needed to achieve higher precisions and performances. It is our goal to demonstrate whether Shannon entropy and Ratio are useful in order to improve previous algorithms.
\noindent

\begin{algorithm}[H]
\caption{Algorithm: CNN Classifier Initialization}
\begin{algorithmic}[1]
\Procedure{init}{vocab\_size, token\_dim, num\_classes}
    \State \textbf{input:} vocab\_size, token\_dim, num\_classes
    \State embedding $\gets$ EmbeddingLayer(vocab\_size, token\_dim, padding\_idx=4)
    \State conv1 $\gets$ Conv1D(in\_channels=token\_dim, out\_channels=16, kernel\_size=3, stride=1, padding=1)
    \State conv2 $\gets$ Conv1D(in\_channels=16, out\_channels=16, kernel\_size=3, stride=1, padding=1)
    \State pool $\gets$ AdaptiveMaxPool1D(1)
    \State dropout $\gets$ Dropout(MODEL\_DROPOUT)
    \State relu $\gets$ ReLU()
    \State input\_dim $\gets$ 16
    \State fc\_final $\gets$ FullyConnected(input\_dim, num\_classes)
\EndProcedure
\end{algorithmic}
\end{algorithm}

\noindent
\textbf{Algorithm: CNN Classifier Initialization.}\\
This algorithm initializes the CNN classifier by setting up the embedding layer, convolutional layers, pooling and dropout.

\begin{algorithm}[H]
\caption{CNN Forward Pass}
\begin{algorithmic}[1]
\Procedure{forwardPass}{input\_sequence}
    \State $x \gets$ EmbeddingLayer(input\_sequence)
    \State $x \gets$ Transpose($x$) \Comment{Adjust dimensions for convolution}
    \State $x \gets$ Convolutional layers with ReLU activations
    \State $x \gets$ AdaptiveMaxPool1D($x$)
    \State $x \gets$ Dropout($x$)
    \State \Return FullyConnected($x$)
\EndProcedure
\end{algorithmic}
\end{algorithm}

\noindent
\textbf{Algorithm: CNN Forward Pass.}\\
This algorithm embeds the raw viral single nucleotide‑DNA tokens into vector form (or trinucleotide‑DNA tokens for mammal genes), transposes the resulting tensor to match the expected convolutional input format, applies convolutional filters with ReLU activations to extract local sequence patterns, reduces dimensionality via adaptive pooling, and finally returns classification logits through a fully connected layer.

\vspace{1em}

\subsection{Comparisons between the Models}

In our implementation, DNA sequences are not tokenized into trinucleotides, in order to deploy even less memory. Instead, each nucleotide is mapped to an integer in \(\{0,1,2,3,4\}\) (corresponding to \(A\), \(C\), \(G\), \(T\) and the padding symbol). For the CNN, the vocabulary size is therefore 5. Two configurations have been considered: one with a token dimension of 8 and one with a token dimension of 256. The parameter counts are computed as follows.

For the configuration with token dimension 8, the embedding layer has \(5 \times 8 = 40\) parameters. The first convolutional layer has \(8 \times 16 \times 3 + 16 = 400\) parameters. The second convolutional layer has \(16 \times 16 \times 3 + 16 = 784\) parameters. The final fully connected layer has \(16 \times 6 + 6 = 102\) parameters. The total is \(40+400+784+102=1326\) parameters.

For the configuration with token dimension 256, the embedding layer has \(5 \times 256 = 1280\) parameters. The first convolutional layer has \(256 \times 16 \times 3 + 16 = 12304\) parameters. The second convolutional layer has \(16 \times 16 \times 3 + 16 = 784\) parameters. The final fully connected layer has \(16 \times 6 + 6 = 102\) parameters. The total is \(1280+12304+784+102=14470\) parameters.

\begin{table}[h!]
\centering
\footnotesize
\renewcommand{\arraystretch}{1.1}
\resizebox{0.815\textwidth}{!}{%
\begin{tabular}{|l|l|c|}
\hline
\textbf{Configuration} & \textbf{Parameter Calculation} & \textbf{Total} \\
\hline
Token Dim = 8 & 
\begin{tabular}[c]{@{}l@{}}Embedding: \(5\times8=40\)\\ CNN Layers: \\[1mm]
\quad Conv1: \(8\times16\times3+16=400\)\\[1mm]
\quad Conv2: \(16\times16\times3+16=784\)\\[1mm]
FC Final: \(16\times6+6=102\)\\[1mm]
Total: \(40+400+784+102=1326\)
\end{tabular} 
& \(1326\) \\
\hline
Token Dim = 256  & 
\begin{tabular}[c]{@{}l@{}}Embedding: \(5\times256=1280\)\\ CNN Layers: \\[1mm]
\quad Conv1: \(256\times16\times3+16=12304\)\\[1mm]
\quad Conv2: \(16\times16\times3+16=784\)\\[1mm]
FC Final: \(16\times6+6=102\)\\[1mm]
Total: \(1280+12304+784+102=14470\)
\end{tabular} 
& \(14470\) \\
\hline
\end{tabular}%
}
\caption{Comparison of the number of parameters in the CNNs. In our implementation, the DNA sequence is represented using a vocabulary of 5 (one per nucleotide plus padding), with two configurations considered: one with token dimension 8 and one with token dimension 256.}
\end{table}

\section{Results}

\subsection{Viral Genes}

Two series of tests were conducted using the benchmark code with the implementation of several CNNs using the dataset of viral genes from \cite{ref14}. A total of \textbf{five} CNNs were implemented based on the previously described algorithms: for both the configuration with token dimension 8 and that with token dimension 256, we implemented one model using data augmentation via Random Crop, one using data augmentation via R‑Based Crop, one using data augmentation via Entropy-based Crop, one using data augmentation via Kolmogorov Complexity algorithm, one without data augmentation. For the token‑dim \(=256\) model, we only used 28 sequences for training and 13 for initial validation, then tested on 400 other sequences 40 times to get meaningful average test accuracies.\footnote{For each of the 40 runs, training was re‑initialised from scratch with \emph{different} random seeds, ensuring that the model never saw the test set during optimisation.} For the token‑dim \(=8\) model we used the whole viral training set (1320 sequences). Here are the results of the benchmarks:

\begin{table}[H]
\centering
\caption{CNN results (Token Dim = 256, training sequences = 28, development sequences = 13, T = 22, n = 1; macro-averaged metrics over 40 iterations) on \texttt{test\_set.csv} (400 sequences)}
\begin{tabular}{l c c c}
\hline
\textbf{Variant}               & \textbf{Accuracy}      & \textbf{Recall}        & \textbf{F1}            \\
                               & (mean $\pm$ std)           & (mean $\pm$ std)           & (mean $\pm$ std)           \\
\hline
Random Crop                    & 0.242 $\pm$ 0.034          & 0.265 $\pm$ 0.037          & 0.226 $\pm$ 0.039          \\
Ratio-based Crop               & 0.859 $\pm$ 0.048          & 0.875 $\pm$ 0.044          & 0.865 $\pm$ 0.045          \\
Entropy-based Crop             & 0.747 $\pm$ 0.069          & 0.770 $\pm$ 0.060          & 0.752 $\pm$ 0.065          \\
Kolmogorov-based Crop          & 0.393 $\pm$ 0.042          & 0.410 $\pm$ 0.042          & 0.380 $\pm$ 0.041          \\
No augmentation (basic crop)   & 0.236 $\pm$ 0.033          & 0.256 $\pm$ 0.032          & 0.219 $\pm$ 0.034          \\
\hline
\end{tabular}
\end{table}

Unless otherwise stated, the 95\% margin of error (MOE) is computed as
\begin{equation}
\mathrm{MOE} \;=\; 1.96 \,\frac{\sigma}{\sqrt{R}},
\end{equation}
where \(\sigma\) is the standard deviation across the \(R=40\) random restarts.

\begin{table}[H]
\centering
\caption{CNN results (Token Dim = 8, training sequences = 1320, development sequences = 180, T = 22, n = 1; macro-averaged metrics over 40 iterations) on \texttt{test\_set.csv} (400 sequences)}
\begin{tabular}{l c c c}
\hline
\textbf{Variant}             & \textbf{Accuracy}         & \textbf{Recall}            & \textbf{F1}                 \\
                              & (mean\,$\pm$\,std)        & (mean\,$\pm$\,std)         & (mean\,$\pm$\,std)          \\
\hline
Entropy-based Crop           & 0.893\,$\pm$\,0.036       & 0.899\,$\pm$\,0.032        & 0.893\,$\pm$\,0.035         \\
Ratio-based Crop             & 0.926\,$\pm$\,0.026       & 0.928\,$\pm$\,0.025        & 0.924\,$\pm$\,0.026         \\
Random Crop                  & 0.468\,$\pm$\,0.040       & 0.481\,$\pm$\,0.043        & 0.461\,$\pm$\,0.040         \\
Kolmogorov-based Crop        & 0.625\,$\pm$\,0.029       & 0.640\,$\pm$\,0.029        & 0.620\,$\pm$\,0.029         \\
No augmentation (basic crop) & 0.466\,$\pm$\,0.034       & 0.480\,$\pm$\,0.035        & 0.458\,$\pm$\,0.033         \\
\hline
\end{tabular}
\end{table}

The 95\% confidence-interval margin of error for accuracy is approximately $\pm1.1\%$ for the entropy-based method, $\pm0.8\%$ for ratio-based, and around $\pm1\%$ for the remaining variants.

\subsection{Human Genes with Repetitions}

For this other dataset~\cite{ref15} we only used the token‑dim \(=256\) configuration because of the scarcity of the data: only 11 sequences associated with expansion pathologies, while 13 not associated. A total of 24 sequences, split into 16 for training and 8 for testing over 40 iterations to get sound average accuracies.\footnote{Each of the 40 runs started from randomly initialised weights and did not reuse the test sequences during training.} 
Here we implemented the same cropping techniques as in the previous set of tests:

\begin{table}[H]
\centering
\caption{CNN results (Token Dim = 256, training sequences = 16, test sequences = 8, T = 98, n = 2; macro-averaged metrics over 40 iterations)}
\begin{tabular}{l c c c}
\hline
\textbf{Variant}             & \textbf{Accuracy}       & \textbf{Recall}          & \textbf{F1}               \\
                             & (mean $\pm$ std)            & (mean $\pm$ std)             & (mean $\pm$ std)              \\
\hline
Random Crop                  & 0.559 $\pm$ 0.068           & 0.559 $\pm$ 0.068            & 0.535 $\pm$ 0.073             \\
Ratio-based Crop             & 0.741 $\pm$ 0.047           & 0.741 $\pm$ 0.047            & 0.731 $\pm$ 0.051             \\
Entropy-based Crop           & 0.666 $\pm$ 0.046           & 0.666 $\pm$ 0.046            & 0.653 $\pm$ 0.048             \\
Kolmogorov-based Crop        & 0.569 $\pm$ 0.050           & 0.569 $\pm$ 0.050            & 0.551 $\pm$ 0.053             \\
No augmentation (basic crop) & 0.584 $\pm$ 0.050           & 0.584 $\pm$ 0.050            & 0.564 $\pm$ 0.053             \\
\hline
\end{tabular}
\end{table}

\noindent

The 95\% margin of error for the accuracy estimates ranges from approximately $\pm$2.1 \% for Random Crop to $\pm$1.5 \% Ratio-based Crop and $\pm$1.4\% for Entropy-based Crop, with intermediate values of $\pm$1.6\% for both Kolmogorov-based Crop and No augmentation.

\section{Discussion}
\subsection{Viral Genes}
Our tests have highlighted that data augmentation techniques based on the introduction of the parameter \(R\), which guides the data augmentation, \textbf{consistently} improve performance. The accuracy achieved with local pattern extractions neural network models such as CNNs on the viral DNA dataset, trained on only 41 sequences, is remarkably high (86.1\%), suggesting important potential applications in the study of cases where the number of available samples is extremely limited. 

Furthermore, the accuracy obtained (92.4\%) by the model with only 1326 parameters, occupying less than 10 kilobytes in \texttt{FP32} format, suggests practical implementations of very small neural networks capable of performing inference on devices with minimal RAM (only a few megabytes). This opens the door to the development of compact genetic self-testing devices.

These results would not be possible without the entropy distribution calculation discussed earlier and the formal introduction of the parameter \(R\). In fact, the other data augmentation techniques are clearly not as precise as our novel method based on \(R\) (although it is to be noted that the n-gram entropy cropping is novel too). This also suggests profound biological and genetic implications, where our entropy, given its effectiveness in distinguishing different gene types, as demonstrated in the benchmark, plays a fundamental role in decoding the genetic information contained in DNA.

We can explain these consistent improvements in performance by analysing the distribution of sequences from the viral gene dataset (\(training\_set.csv\)) on a plane. We define a plane where the vertical axis represents the GC content (commonly used in these classification problems) and the horizontal axis represents:

1. The Kolmogorov Complexity calculated on zlib compression.\\

2. The classical Shannon Entropy on the full sequences with \(n=1\).\\

3. The Ratio of the sequences with \(T=22, n=1\).\\

We obtain the following distributions for the six classes:

\begin{figure}[H]
    \centering
    \includegraphics[width=1.00\textwidth]{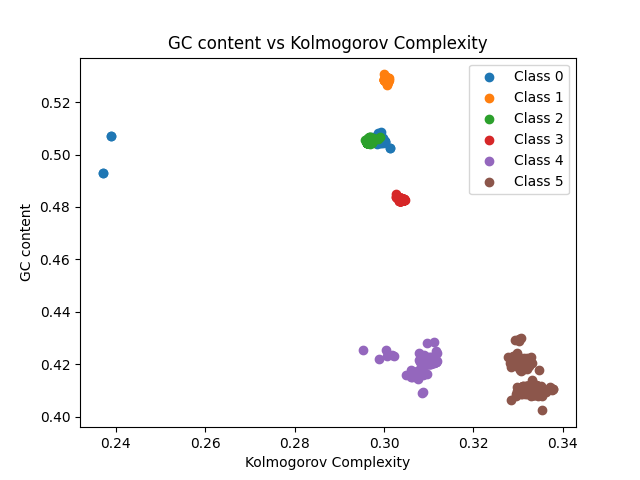}
    \caption{Distribution of the six viral classes \(training\_set.csv\) on the GC content--Kolmogorov Complexity (based on zlib compression) plane.}
    \label{fig:3.1}
\end{figure}

\begin{figure}[H]
    \centering
    \includegraphics[width=1.00\textwidth]{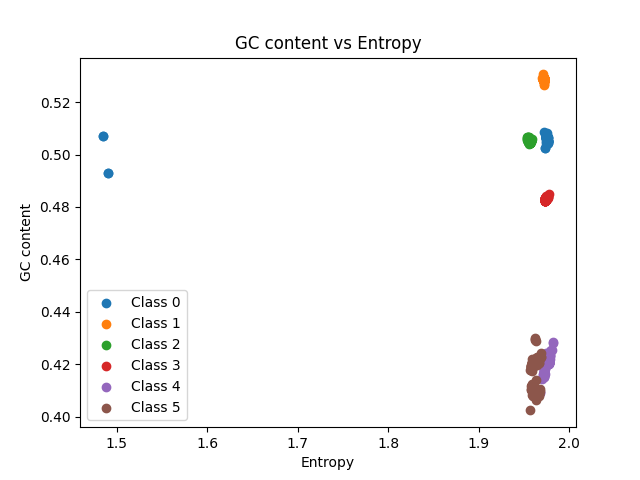}
    \caption{Distribution of the six viral classes \(training\_set.csv\) on the GC content--Shannon Entropy \(n=1\) plane.}
    \label{fig:3.2}
\end{figure}

\begin{figure}[H]
    \centering
    \includegraphics[width=1.00\textwidth]{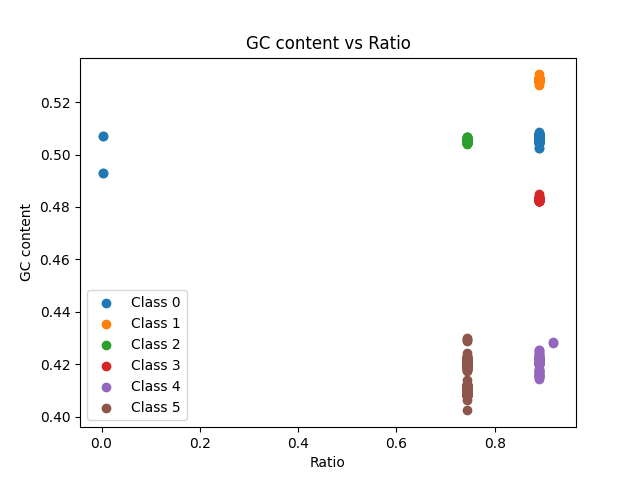}
    \caption{Distribution of the six viral classes \(training\_set.csv\) on the GC content--Ratio with \(T=22, n=1\) plane.}
    \label{fig:3.3}
\end{figure}

As we can see from the distributions, we have overlapping classes for the Kolmogorov Complexity (class 0 and class 2) and the Shannon Entropy (class 4 and class 5). For the R-based distribution of the DNA sequences, we have almost perfect discrimination among sequences with different \(R\) value and same GC Content value. Furthermore, the difference between classes with similar amount of GC content is larger in the case of Ratio (tenths), while the other methods only differentiate classes in the order of cents.

It is logical to think that cropping algorithms based on \(R\) will not only be more precise, but easier to fine-tune, being more stable on changes in the parameters of the Neural Network.

\subsection{Human Genes with Repetitions}

The value of these results is further supported by the sound results obtained on the dataset of human DNA sequences subject to expansion (73.4\% with only 24 sequences using the R-based crop).

We can plot similar distributions over this dataset, with similar results to the previous ones, using \(n=2\) for Shannon Entropy and \(T=98,n=2\) for Ratio:

\begin{figure}[H]
    \centering
    \includegraphics[width=1.00\textwidth]{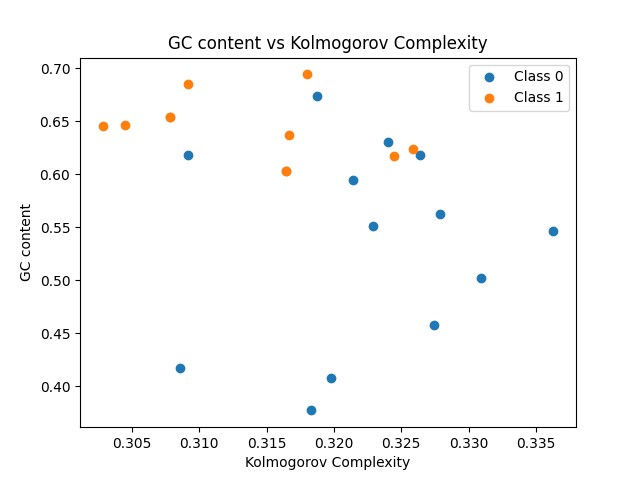}
    \caption{Distribution of the two classes (pathological and not pathological) \(human\_training\_set.csv\) on the GC content--Kolmogorov Complexity plane.}
    \label{fig:4.1}
\end{figure}

\begin{figure}[H]
    \centering
    \includegraphics[width=1.00\textwidth]{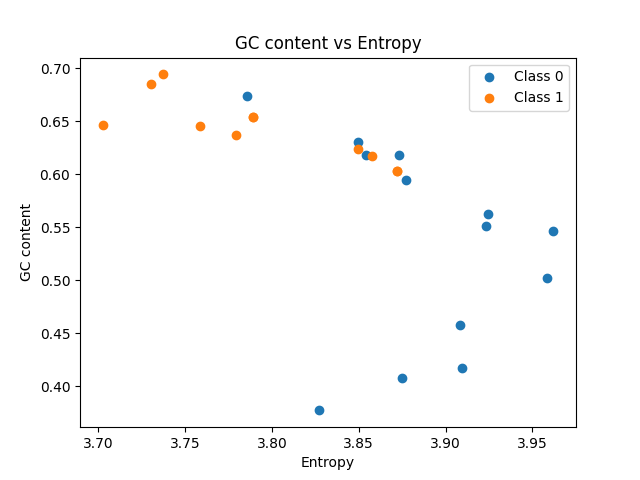}
    \caption{Distribution of the two classes (pathological and not pathological) \(human\_training\_set.csv\) on the GC content-- Shannon Entropy \(n=2\) plane.}
    \label{fig:4.2}
\end{figure}

\begin{figure}[H]
    \centering
    \includegraphics[width=1.00\textwidth]{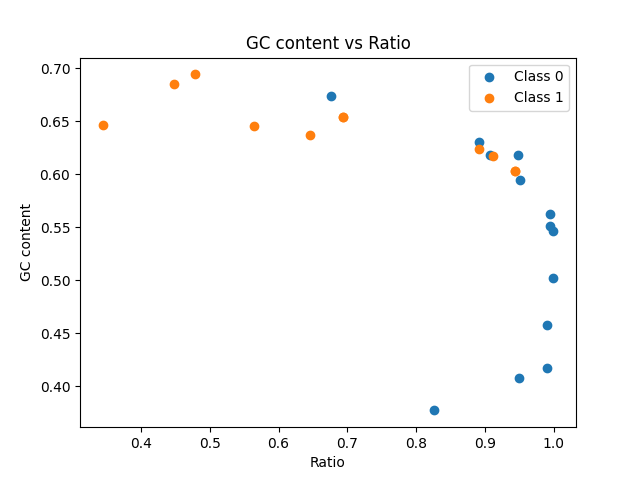}
    \caption{Distribution of the two classes (pathological and not pathological) \(human\_training\_set.csv\) on the GC content--Ratio \(T=98, n=2\) plane.}
    \label{fig:4.3}
\end{figure}

As we can see, while Kolmogorov Complexity gives less precise distinction between the two classes, both Shannon Entropy and Ratio give similar results, but the parameter \(R\) provides again a larger more precise representation and a larger gap between the two main clusters.

Hence, we can understand that, often, \(R\) tends to form a distinct cluster (with only negligible exceptions). This greatly facilitates the task of classifying DNA sequences.

\subsubsection{On information content.}
At fixed \((T,n)\), \(R\) is a monotone transform of the block entropy and therefore does not add raw information beyond \(S\); its benefit is that it normalizes values with respect to the full combinatorial distribution, making them directly comparable across sequences.

\section{Conclusions}

Our work shows that it is possible to introduce a new perspective on Shannon Entropy for DNA, based on the novel parameter \(R\), defined as the ratio that locates a target sequence within the distribution of all possible sequences of the same length. This is achieved by subdividing the sequence into fixed-length subsequences and non-overlapping \(n\)-tuples. By doing so, and thereby avoiding the limitations of classical Shannon Entropy (such as the tendency of entropy values to approach 2 for most sequences) we developed cropping algorithms based on the parameter \(R\) that \textbf{consistently} increase the precision of CNNs on small datasets or with limited parameters. Even though these results are preliminary and obtained only on two datasets one of viral classes and another of human genes subject to polynucleotide expansion (which we believe at this moment is the actual limitation of this novel method), we expect entropy to play an increasingly crucial role in the interpretation of genetic information, with all the resulting benefits in the medical and biological fields. Future work could include implementing these techniques in the field of pathological polynucleotide repetitions or creating compact devices that run very lightweight neural networks for classification. We look forward to applying these methods in the study of rare genetic diseases, which we believe could be the major impact of our work.

\section{Data availability}

The code used for the benchmarks is available in the GitHub repository \cite{ref15} along with a dataset of 24 human DNA sequences subject to polynucleotide expansion (pathological and not). The reference dataset for viral DNA sequences is taken from \cite{ref14}, consisting of 6 viral classes, with a training CSV file containing 1320 sequences, a development set with 180 sequences, and a test dataset with 400 sequences. The codes were written in Python using the PyTorch library\cite{ref34} that allows GPU hardware acceleration with CUDA\cite{ref35}. 
Values are macro-averages over 40 random restarts.\\
The machine used for running the code has the following components:\\
2 Intel Xeon Gold 6338\\
1024 GB DDR4 3200 ECC Registered RAM\\
Intel C621A Chipset\\
16 TB SATA/600 HDD (7200 rpm)\\
2.5" 7.68 TB NVMe PCIe 4.0 SSD\\
4 GPU RTX A5000 with 24 GB VRAM (230 W each)\\
OS: Ubuntu Server 24.04

\section{Competing interests}
No competing interest is declared.

\section{Acknowledgments}

The authors gratefully acknowledge the Department of Biology and Earth Science at the University of Calabria for providing access to their high-performance workstation, which was instrumental in conducting the computational experiments.
The authors thank Dr. Gabriel Antonio Videtta for his valuable suggestions in proving Theorem 1.


\begin{thebibliography}{35}

\bibitem[1]{ref1}
Watson, J.~D.; Crick, F.~H.~C.\\
\textit{Molecular Structure of Nucleic Acids: A Structure for Deoxyribose Nucleic Acid}.\\
\textit{Nature} \textbf{1953,} \textit{171}, 737--738.\\
DOI: \href{https://doi.org/10.1038/171737a0}{10.1038/171737a0}.

\bibitem[2]{ref2}
Nirenberg, M.~W.; Matthaei, J.~H.\\
\textit{The Dependence of Cell-Free Protein Synthesis on Naturally Occurring or Synthetic Polyribonucleotides}.\\
\textit{Proc. Natl. Acad. Sci. USA} \textbf{1961,} \textit{47}(10), 1588--1602.\\
DOI: \href{https://doi.org/10.1073/pnas.47.10.1588}{10.1073/pnas.47.10.1588}.

\bibitem[3]{ref3}
Gatlin, L.~L.\\
\textit{The Information Content of DNA}.\\
\textit{Journal of Theoretical Biology} \textbf{1966,} \textit{10}(2), 281--300.\\
DOI: \href{https://doi.org/10.1016/0022-5193(66)90127-5}{10.1016/0022-5193(66)90127-5}.

\bibitem[4]{ref4}
Schneider, T.~D.; Stormo, G.~D.; Gold, L.; Ehrenfeucht, A.\\
\textit{Information Content of Binding Sites on Nucleotide Sequences}.\\
\textit{Journal of Molecular Biology} \textbf{1986,} \textit{188}(3), 415--431.\\
DOI: \href{https://doi.org/10.1016/0022-2836(86)90165-8}{10.1016/0022-2836(86)90165-8}.

\bibitem[5]{ref5}
Shannon, C.~E.\\
\textit{A Mathematical Theory of Communication}.\\
\textit{Bell Syst. Tech. J.} \textbf{1948,} \textit{27}, 379--423 (and 623--656).\\
DOI: \href{https://doi.org/10.1002/j.1538-7305.1948.tb01338.x}{10.1002/j.1538-7305.1948.tb01338.x}.

\bibitem[6]{ref6}
Koslicki, D.\\
\textit{Topological Entropy of DNA Sequences}.\\
\textit{Bioinformatics} \textbf{2011,} \textit{27}(8), 1061--1067.\\
DOI: \href{https://doi.org/10.1093/bioinformatics/btr077}{10.1093/bioinformatics/btr077}.

\bibitem[7]{ref7}
Schmitt, A.~O.; Herzel, H.\\
\textit{Estimating the Entropy of DNA Sequences}.\\
\textit{Journal of Theoretical Biology} \textbf{1997,} \textit{188}(3), 369--377.\\
DOI: \href{https://doi.org/10.1006/jtbi.1997.0493}{10.1006/jtbi.1997.0493}.

\bibitem[8]{ref8}
Angermueller, C.; Pärnamaa, T.; Parts, L.; Stegle, O.\\
\textit{Deep Learning for Computational Biology}.\\
\textit{Mol. Syst. Biol.} \textbf{2016,} \textit{12}, 878.\\
DOI: \href{https://doi.org/10.15252/msb.20156651}{10.15252/msb.20156651}.

\bibitem[9]{ref9}
Eraslan, G.; Avsec, \v{Z}.; Gagneur, J.; et al.\\
\textit{Deep Learning: New Computational Modelling Techniques for Genomics}.\\
\textit{Nat. Rev. Genet.} \textbf{2019,} \textit{20}, 389--403.\\
DOI: \href{https://doi.org/10.1038/s41576-019-0122-6}{10.1038/s41576-019-0122-6}.

\bibitem[10]{ref10}
Ching, T.; Himmelstein, D.~S.; Beaulieu-Jones, B.~K.; Kalinin, A.~A.; Do, B.~T.; Way, G.~P.; et al.\\
\textit{Opportunities and Obstacles for Deep Learning in Biology and Medicine}.\\
\textit{J. R. Soc. Interface} \textbf{2018,} \textit{15}(141), 20170387.\\
DOI: \href{https://doi.org/10.1098/rsif.2017.0387}{10.1098/rsif.2017.0387}.

\bibitem[11]{ref11}
Yao, M.\\
\textit{Fuzzy Entropy Based Combined Learning Algorithm for Neural Networks}.\\
\textit{J. Syst. Eng. Electron.} \textbf{1996,} \textit{7}(1), 15--22.

\bibitem[12]{ref12}
Alipanahi, B.; Delong, A.; Weirauch, M.; et al.\\
\textit{Predicting the Sequence Specificities of DNA- and RNA-Binding Proteins by Deep Learning}.\\
\textit{Nat. Biotechnol.} \textbf{2015,} \textit{33}, 831--838.\\
DOI: \href{https://doi.org/10.1038/nbt.3300}{10.1038/nbt.3300}.

\bibitem[13]{ref13}
Ji, Y.; Zhou, Z.; Liu, H.; Davuluri, R.~V.\\
\textit{DNABERT: Pre-trained Bidirectional Encoder Representations from Transformers Model for DNA-Language in Genome}.\\
\textit{Bioinformatics} \textbf{2021,} \textit{37}(15), 2112--2120.\\
DOI: \href{https://doi.org/10.1093/bioinformatics/btab083}{10.1093/bioinformatics/btab083}.

\bibitem[14]{ref14}
Zolfaghari, A.\\
\textit{DNA Sequence Classification Dataset}.\\
GitHub Repository, 2025.\\
Available at: \url{https://github.com/arminZolfaghari/DNA-Sequence-Classification/tree/main/Dataset} (Accessed: 1 March 2025).

\bibitem[15]{ref15}
Pastore, E. P.\\
\textit{Benchmark codes for CNN DNA-classifiers that use different complexity-based cropping techniques}.\\
Zenodo Repository, 2025.\\
Available at: \url{https://doi.org/10.5281/zenodo.15399359} (Accessed: 13 May 2025).



\bibitem[16]{ref16}
Tenreiro Machado, J.~A.; Costa, A.~C.; Quelhas, M.~D.\\
\textit{Entropy Analysis of the DNA Code Dynamics in Human Chromosomes}.\\
\textit{Computers \& Mathematics with Applications} \textbf{2011,} \textit{62}(3), 1612--1617.\\
DOI: \href{https://doi.org/10.1016/j.camwa.2011.03.005}{10.1016/j.camwa.2011.03.005}.

\bibitem[17]{ref17}
Jin, S.; Tan, R.; Jiang, Q.; Xu, L.; Peng, J.; Wang, Y.; Wang, Y.\\
\textit{A Generalized Topological Entropy for Analyzing the Complexity of DNA Sequences}.\\
\textit{PLoS ONE} \textbf{2014,} 9, e88519.\\
DOI: \href{https://doi.org/10.1371/journal.pone.0088519}{10.1371/journal.pone.0088519}.

\bibitem[18]{ref18}
Loewenstern, D.; Yianilos, P.~N.\\
\textit{Significantly Lower Entropy Estimates for Natural DNA Sequences}.\\
\textit{J. Comput. Biol.} \textbf{1999,} \textit{6}(1), 125--142.\\
DOI: \href{https://doi.org/10.1089/cmb.1999.6.125}{10.1089/cmb.1999.6.125}.

\bibitem[19]{ref19}
Dorier.\\
\textit{A General Outline of the Genesis of Vector Space Theory}.\\
\textit{Historia Mathematica} \textbf{1995,} \textit{22}(3), 227--261.\\
DOI: \href{https://doi.org/10.1006/hmat.1995.1024}{10.1006/hmat.1995.1024}.

\bibitem[20]{ref20}
Marcais, G.; Kingsford, C.\\
\textit{A Fast, Lock-Free Approach for Efficient Parallel Counting of Occurrences of k-mers}.\\
\textit{Bioinformatics} \textbf{2011,} \textit{27}(6), 764--770.\\
DOI: \href{https://doi.org/10.1093/bioinformatics/btr011}{10.1093/bioinformatics/btr011}.

\bibitem[21]{ref21}
Berkes, I.; Philipp, W.; Tichy, R.\\
\textit{Entropy Conditions for Subsequences of Random Variables with Applications to Empirical Processes}.\\
\textit{Monatshefte f\"ur Mathematik} \textbf{2008,} \textit{153}, 183--204.\\
DOI: \href{https://doi.org/10.1007/s00605-007-0519-8}{10.1007/s00605-007-0519-8}.

\bibitem[22]{ref22}
Farach-Colton, M.; Noordewier, M.; Savari, S.; Shepp, L.; Wyner, A.; Ziv, J.\\
\textit{On the Entropy of DNA: Algorithms and Measurements Based on Memory and Rapid Convergence}.\\
In: \textit{Proc. Sixth Annu. ACM-SIAM Symp. Discrete Algorithms} \textbf{1995,} 48--57.\\
DOI: \href{https://doi.org/10.1145/313651.313662}{10.1145/313651.313662}.

\bibitem[23]{ref23}
Biswas, S.; Sarkar, B.~K.\\
\textit{Entropy of DNA Sequences as Similarity Index for Various SARS-CoV-2 Virus Strains}.\\
In: \textit{Advances in Medical Physics and Healthcare Engineering}, Lecture Notes in Bioengineering, edited by Mukherjee, M.; Mandal, J.; Bhattacharyya, S.; Huck, C.; Biswas, S., Springer, Singapore, 2021.\\
DOI: \href{https://doi.org/10.1007/978-981-33-6915-3_51}{10.1007/978-981-33-6915-3\_51}.

\bibitem[24]{ref24}
Janson, S.; Lonardi, S.; Szpankowski, W.\\
\textit{On the Average Sequence Complexity}.\\
In: \textit{Combinatorial Pattern Matching (CPM 2004), Lecture Notes in Computer Science}, edited by Sahinalp, S.~C.; Muthukrishnan, S.; Dogrusoz, U., vol. 3109, pp. 63--76, Springer, Berlin, Heidelberg, 2004.\\
DOI: \href{https://doi.org/10.1007/978-3-540-27801-6_6}{10.1007/978-3-540-27801-6\_6}.

\bibitem[25]{ref25}
Ebeling, W.; Jimenez-Montano, M.; Pohl, T.\\
\textit{Entropy and Complexity of Sequences}.\\
In: Karmeshu (Ed.), \textit{Entropy Measures, Maximum Entropy Principle and Emerging Applications}, Studies in Fuzziness and Soft Computing, vol. 119, pp. 233--262, Springer, Berlin, Heidelberg, 2003.\\
DOI: \href{https://doi.org/10.1007/978-3-540-36212-8_11}{10.1007/978-3-540-36212-8\_11}.

\bibitem[26]{ref26}
Tsallis, C.\\
\textit{Possible Generalization of Boltzmann--Gibbs Statistics}.\\
\textit{J. Stat. Phys.} \textbf{1988,} \textit{52}(1--2), 479--487.\\
DOI: \href{https://doi.org/10.1007/BF01016429}{10.1007/BF01016429}.

\bibitem[27]{ref27}
Huang, W.; Maass, A.; Ye, X.\\
\textit{Sequence Entropy Pairs and Complexity Pairs for a Measure}.\\
\textit{Ann. Inst. Fourier} \textbf{2004,} \textit{54}(4), 1005--1028.\\
DOI: \href{https://doi.org/10.5802/aif.2041}{10.5802/aif.2041}.

\bibitem[28]{ref28}
Bakouch, H.~S.; Aghababaei Jazi, M.; Nadarajah, S.\\
\textit{A New Discrete Distribution}.\\
\textit{Statistics: A Journal of Theoretical and Applied Statistics} \textbf{2012,} \textit{48}(1), 1--41.\\
DOI: \href{https://doi.org/10.1080/02331888.2012.716677}{10.1080/02331888.2012.716677}.

\bibitem[29]{ref29}
Stanley, R.~P.\\
\textit{Enumerative Combinatorics: Volume 1}.\\
Cambridge University Press, 2011.

\bibitem[30]{ref30}
Debnath, L.\\
\textit{Srinivasa Ramanujan (1887--1920) and the Theory of Partitions of Numbers and Statistical Mechanics: A Centennial Tribute}.\\
\textit{Int. J. Math. Math. Sci.} \textbf{1987.}\\
DOI: \href{https://doi.org/10.1155/S0161171287000772}{10.1155/S0161171287000772}.

\bibitem[31]{ref31}
Ho, S.-W.; Verd\'{u}, S.\\
\textit{Convexity/Concavity of R\'{e}nyi Entropy and $\alpha$-Mutual Information}.\\
In: \textit{2015 IEEE International Symposium on Information Theory (ISIT)}, Hong Kong, China, 2015, pp. 745--749.\\
DOI: \href{https://doi.org/10.1109/ISIT.2015.7282554}{10.1109/ISIT.2015.7282554}.

\bibitem[32]{ref32}
Lockwood, E.\\
\textit{The Strategy of Solving Smaller, Similar Problems in the Context of Combinatorial Enumeration}.\\
\textit{Int. J. Res. Undergrad. Math. Ed.} \textbf{2015,} 1, 339--362.\\
DOI: \href{https://doi.org/10.1007/s40753-015-0016-8}{10.1007/s40753-015-0016-8}.

\bibitem[33]{ref33}
Li, W.; Freudenberg, J.; Miramontes, P.\\
\textit{Diminishing Return for Increased Mappability with Longer Sequencing Reads: Implications of the k-mer Distributions in the Human Genome}.\\
\textit{BMC Bioinformatics} \textbf{2014,} 15, 2.\\
DOI: \href{https://doi.org/10.1186/1471-2105-15-2}{10.1186/1471-2105-15-2}.

\bibitem[34]{ref34}
Paszke, A.; Gross, S.; Massa, F.; Lerer, A.; Bradbury, J.; Chanan, G.; Killeen, T.; Lin, Z.; Gimelshein, N.; Antiga, L.; et al.\\
\textit{PyTorch: An Imperative Style, High-Performance Deep Learning Library}.\\
In: \textit{Advances in Neural Information Processing Systems 32}, 2019, pp. 8024--8035.\\
DOI: \href{https://doi.org/10.48550/arXiv.1912.01703}{10.48550/arXiv.1912.01703}.

\bibitem[35]{ref35}
\textit{CUDA Compatible GPU Cards as Efficient Hardware Accelerators for Smith--Waterman Sequence Alignment}.\\
\textit{BMC Bioinformatics} \textbf{2008,} 9(S2), S10.\\
DOI: \href{https://doi.org/10.1186/1471-2105-9-S2-S10}{10.1186/1471-2105-9-S2-S10}.

\end{thebibliography}
\end{document}